\documentclass[11pt,titlepage]{amsart}
\usepackage{amssymb,amsmath,amsthm,amsfonts}
\usepackage{graphicx}
\usepackage{array}
\usepackage{psfrag}
\usepackage{subfigure}

\newtheorem{theorem}{Theorem}[section]

\newtheorem{algorithm}[theorem]{Algorithm}

\newcommand{\sumi}[1]{\sum_{m=1}^{{#1}}}
\newcommand{\sumij}[1]{\sum_{m,n=1}^{{#1}}}
\newcommand{\pd}[2]{\frac{\partial f}{\partial {#1}^{({#2})}}}
\newcommand{\pds}[2]{\frac{\partial^2 f}{\partial {{#1}^{({#2})}}^2}}
\newcommand{\pdm}[4]{\frac{\partial^2 f}{\partial {#1}^{({#2})} \partial {#3}^{({#4})}}}

\newcommand{\pdt}{\frac{\partial f}{\partial t}}
\newcommand{\rh}[2]{\rho_{{#1}{#2}}}
\newcommand{\E}{\mathbb E}
\newcommand{\rml}{\right.\\}
\newcommand{\lml}{&\quad \left.}

\newcommand{\vi}{v^{(m)}}

\begin{document}
\title[]{ Pricing Spread Options under Stochastic Correlation and  Jump-Diffusion Models}
\author{Matthew Cane, Pablo Olivares}
\email{pablo.olivares@ryerson.ca, mcane2ryerson.ca}

\begin{abstract}
This paper examines the problem of pricing spread options under  some models with jumps driven by  Compound Poisson Processes and stochastic volatilities in the form of  Cox-Ingresoll-Ross(CIR) processes.  We derive the characteristic function for two market models featuring joint normally distributed jumps, stochastic volatility, and  different stochastic dependence structures. With the  use of Fast Fourier Transform(FFT) we  accurately compute  spread option prices across a variety of strikes and initial price vectors at a very low computational cost when compared to Monte Carlo pricing methods. We also look at the sensitivities of the prices to the model specifications and find strong dependence on the selection of the jump and stochastic volatility parameters. Our numerical implementation is based on the method developed by Hurd and Zhou (2009).
\end{abstract}
\maketitle

\section{Introduction}
This paper examines the problem of pricing spread options under market models with jumps driven by  correlated Compound Poisson Processes and stochastic volatilities in the form of  CIR processes, using a bivariate inverse Fourier transform method. \\
While the Black-Scholes model, see Black and Scholes (1973), was an important leap in the mathematical modeling of asset prices, its is now well documented that the model fails to capture critical empirical features observed in financial markets. Heston (1993) has extended the former by including a stochastic volatility dynamic to better model the implied volatility smiles and smirks observed in option prices, while several authors have introduced random jumps to capture sudden oscillations in the asset prices, see the pioneer works of Merton (1974)  and Kou and Wang (2004).\\
 On the other hand Bates (1996)  combines both features, i.e. stochastic volatility and jumps, an idea followed more recently in, for example, Thavaneswaran and Singh(2010) and Wang, Shashan and Shenghong (2009). As Cont and Tankov(2003)  note the presence of both jumps and stochastic volatility terms allows for a much greater flexibility in modeling both short and long term smiles. See also Zhanga and Wangb (2013) for a model with stochastic volatility, stochastic interest rates and jumps.\\
 Most of these models can  be naturally extended to a multivariate setting in order to  price derivatives whose value depends on multiple assets such as spreads, basket options, correlation options, quantos, etc.\\
    The pricing of  multi-asset derivatives has been mostly studied under models  assuming continuous trajectories and a constant correlation between the underlying, see for example Carmona and Durrleman (2003)  or Li, Deng and Zhou (2008), using suitable approximations of the corresponding discounted expected payoff.\\
     An extension to a model with stochastic correlation, but without jumps, is considered in  Dempster and Hong (2000)  where spread prices are computed for asset models with stochastic proportional volatilities. On the other hand Hurd and Zhou (2009)  consider spread option pricing when assets follow discontinuous trajectories described by a bivariate Variance Gamma Levy process. Both approaches use FFT, although in a different way.\\
 Since its introduction by  Carr and Madan (1999), the use of FFT methods has become a standard in pricing univariate derivatives. Implementing it on a two dimensional grid determined by the spread payoff still poses some numerical challenges, among them the effects of the grid, truncation integration and damping parameter in  errors when computing prices. Moreover, bivariate FFT implementations under a model with jumps and stochastic correlation seem have not being considered yet. Our findings show that Hurd and Zhou's method can be well adapted to models with both above mentioned features. See also  Kwok Y.K.,  Leung K.S., and Hoi Ying Wong H.Y.(2012) for  discussion about efficiency and implementation. It is worth noticing that in these cases a Monte Carlo technique requires a large number of simulations to achieve a reasonable  accuracy, in particular when a sensitivity analysis is considered.
\\
The organization of the paper is as follows. In section 2, multivariate derivative and Fourier transform methods for pricing options are reviewed, it walks through the derivation of the pricing method outlined by Hurd and Zhou(2009).  In section 3,  the development of market models for the movement of asset prices is presented, as well as the rationale behind our selection of the market models. We then work through the derivation of the characteristic function for both of our market models. Section 4 will look at the numerical results from our implementation of both a Monte Carlo and Fourier transform pricing tool, and examine the sensitivity of our models to the various input parameters. Finally, Section 5 contains a summary.
\section{Multivariate Derivative and the FFT Pricing Method}
 We consider a pair of stochastic processes $S_t=(S^{(1)}_t, S^{(2)}_t)_{0 \leq t \leq T}$  defined on the stochastic basis $(\Omega ,\mathcal{A},\{\mathcal{F}_{t}\}_{t \geq 0},\mathbb{P})$, with the filtration $\{\mathcal{F}_{t}\}_{t \geq 0}$ satisfying the usual conditions. A European  derivative with an arbitrary payoff, $H(S_T)$ at maturity date $T$ is considered.\\
  Assuming the existence of a risk neutral equivalent martingale measure $\mathbb{Q}$ the value (or price) of the derivative at time $t$ is given by:
\begin{equation}\label{eq:fundthm}
V(S_T) = e^{-r_f (T-t)} {\E}_\mathbb{Q} (H (S_T) | \mathcal{F}_{t})
\end{equation}
where $r_f$ is a deterministic risk-free interest rate and $\E_{\mathbb{Q}}$ indicates the expected value under $\mathbb{Q}$.\\ 
While this paper explores only the pricing of spread options, here we present other multivariate derivatives and their payoffs to illustrate a selection of the products that are traded both over-the-counter and through exchanges. The method is easily extended to other multi-asset derivatives such as \textit{exchange contracts}, see Margrabe(1978) for the pricing within a Black-Scholes framework and Cheang and Chiarella(2008)  under a jump-diffusion market model, the case of \textit{basket options, correlation options and quantos}.  \\
Spread options are derivatives whose payoffs are  based on the difference between two asset prices. The payoff of a European Call spread option at time T  is given by:

\begin{equation*}
H(S_T) = (S^{(1)}_T - S_T^{(2)} - K)_+
\end{equation*}
where K is  the strike price, and the notation $x_+=max(x,0)$.\\
Spread options are used both as hedging tools and as speculative instruments. Spread options are widespread in commodity markets, where they can be used to hedge against the conversion or production costs for raw goods. For example the crack spread is based on the difference in prices between refined oil products (such as gasoline) and crude oil and sparks look at the difference between prices of electricity and oil. See for example Hikspoors and  Jaimungal(2006),  or   Hambly, Howison and Kluge(2007).\\
Spread options can also be used as speculative tools, as they allow the purchaser to effectively "bet" on the correlation between the two assets. For example, if an investor believes that the correlation between the assets will decrease (and thus the spread may widen) the investor would long a call, while if they believe a correlation increase will occur (and thus the spread will likely hold at the same level), they should write a call on the underlying assets.\\
Fourier transform methods provide an efficient and widely-used alternative to pricing by Monte Carlo and other numerical methods when the characteristic function of the underlying market model is known. Carr and Madan(1999)  first used the Fourier transform to price European call options, while both Dempster and Hong(2000) and Hurd and Zhou(2009)  derived Fourier Transform methods to price spread options. Eberlien et al.(2009) give an overview of both the univariate and multivariate cases which have been examined so far, and look at pricing options based on the minimum price of a basket of assets.\\
Fourier transform methods rely on the knowledge of the characteristic function for the underlying market model. We denote the characteristic function of a random vector $X = (X^{(1)}, X^{(2)},..., X^{(d)})$ by:
\begin{equation*}
\phi_{X} (u) = \E_{\mathbb{Q}} \left( e^{i u \cdot X} \right)\;\;\text{for}\;\;u \in \mathbb{R}^d
\end{equation*}
where $a \cdot b$ denotes the scalar product of $a$ and $b$.\\
We outline Hurd and Zhou's method in the following terms:\\
For a derivative dependent on two assets, with an arbitrary payoff $H(S_T)$ with Fourier transform $\hat{H}_T (u)$, the price at time $t=0$ of the derivative can be calculated as:
\begin{equation}\label{eq:fourierPrice}
V_T(S_0) =  \frac{e^{-r_f T}}{(2 \pi)^2} \iint \limits_{\mathbb{R}^2 + i \epsilon} e^{i u \cdot S_0} \phi_{S_T} (u) \hat{H}_T(u) d^2 u
\end{equation}
where we note that since the increments in $S_t - S_0$ is independent of $S_0$ (as it will be in all the models we analyze), we can write:
\begin{equation*}
\E_{\mathbb{Q}} (e^{i u \cdot S_T} | S_0) = e^{i u \cdot S_0} \phi_{S_T} (u)
\end{equation*}
It begins by discretizing the complex domain as:
\begin{equation*}
\Gamma = \{ u(k) = (u(k_1), u(k_2)) | k = (k_1,k_2), \in \{0,1,...,(N-1)\}^2 \}
\end{equation*}
where $u_i (k_i) = -\bar{u} + k_i \eta$ over N points, with $\bar{u} = \frac{N \eta}{2}$ being the truncated end-points for our numerical integration. Based on our choices for $N$, $\eta$ and $\bar{u}$, we can discretize the real domain as:
\begin{equation*}
\Gamma^* = \{ x(l) = x(l_1), x(l_2) | l = (l_1,l_2), \in {0,...,N-1}^2 \}, x_i (l_i) = -\bar{x} + l_i \eta^*
\end{equation*}
where $\eta^* = \frac{2 \pi}{N \eta}$ and $\bar{x} = \frac{N \eta^*}{2}$. The value of contract can then be estimated as:
\begin{eqnarray*}\notag
V_T(S_0) & \thicksim (-1)^{l_1 + l_2} e^{-r_f T} {\left(\frac{\eta N}{2 \pi} \right)}^2 e^{-\epsilon \cdot x(l)} \left[ \frac{1}{N^2} \sum_{k_1,k_2=1}^{N-1} e^{\frac{-2 \pi i k \cdot l}{N}} G(k) \right] \\
& = (-1)^{l_1 + l_2} e^{-r_f T} {\left(\frac{\eta N}{2 \pi} \right)}^2 e^{-\epsilon \cdot x(l)} \left[\textrm{ifft2}(G(k)) \right](l)
\end{eqnarray*}
 Here $\textrm{ifft2}(J)$ indicates the two dimensional inverse FFT (or any discrete Fourier transform) of $J$, and $G(k)$ is defined as:
\begin{equation*}
G(k) = (-1)^{k_1 + k_2} \phi_{X_t} (u(k) + i \epsilon) \hat{P} (u + i\epsilon)
\end{equation*}
where  $\hat{P} (u)$ is defined as the Fourier transform of the payoff function. The Fourier transform of a spread option, $\hat{H} (u)$ for the case of $K = 1$ as:
\begin{equation}\label{eq:SpreadPayoff}
\hat{H} (u) = \frac{\Gamma (i(u_1 + u_2) - 1) \Gamma(-i u_2)}{\Gamma(i u_1 + 1)}
\end{equation}
where $\Gamma(a)$ is the complex gamma function defined for $Re(a) > 0$.  \\
This method can be easily extended to the case of $K \neq 1$, $K > 0$, by simply making a change of variables. If we define the spread value when $K = 1$ by:
\begin{equation*}
Spr(S_0^{(1)}, S_0^{(2)},1) = e^{-r_f T} {\E}_{\mathbb Q} \left( {(S_T^{(1)} - S_T^{(2)} - 1)}_+ | S_0^{(1)}, S_0^{(2)} \right)
\end{equation*}
then for the case of $K \neq 1$ we can write the spread price as:
\begin{equation*}
Spr(S_0^{(1)}, S_0^{(2)},K) = e^{-r_f T} {\E}_{\mathbb Q} \left( {(S_T^{(1)} - S_T^{(2)} - K)}_+ | S_0^{(1)}, S_0^{(2)} \right) \\
\end{equation*}
If we make the change of variable, $Y_t^{(m)} = \frac{S_t^{(m)}}{K}$ for $m=1,2,\ldots,m$ then our equation becomes:

\begin{equation*}\begin{split}
 Spr(S_0^{(1)}, S_0^{(2)},K) & =  e^{-r_f T} {\E}_{\mathbb Q} \left( K \left(Y_T^{(1)} - Y_T^{(2)} - 1 \right)_+ | Y_0^{(1)}, Y_0^{(2)} \right) \\
  & = K  Spr(Y_0^{(1)}, Y_0^{(2)},K) = K  Spr \left(\frac{S_0^{(1)}}{K}, \frac{S_0^{(2)}}{K},1 \right)
\end{split}\end{equation*}
We can also take steps to ensure that both of our initial asset prices land on the inverse grid $\Gamma^*$. Standard FFT methods implement a model with equal step sizes of $\eta$ and $\eta^* =  \frac{2 \pi}{N \eta}$ along the x and y axes of the complex and real planes respectively. If we instead specify the step size along each axes of each plane (i.e. $\eta^{(1)}$, $\eta^{(2)}$ and $\eta^{* (1)}$, $\eta^{* (2)}$), we can eliminate the need for any interpolation between grid points. We can additionally specify a minimum integration interval in the complex plane $\overline{u}_{min}$ , and use the algorithm given below to find a step size size with a minimum truncation error and each initial asset price on the grid. The algorithm can be summarize as follows:

\begin{algorithm} \label{alg:grid} Algorithm for Selecting Step Size $\eta^{(m)}$ given $N$, $S_0^{(m)}$ and $\overline{u}_{min}$
\begin{enumerate}
\item Select $N$, $\bar{u}_{min}$.
\item For each asset $m = 1,2$, with initial price $S_0^{(m)}$ set log-price\\
 $X_0^{(m)} =\log S_0^{(m)}/K $, strike price K.
\quad \item Set $\bar{u}_{Test} = \frac{\pi i - N/2}{X_0^{(m)}}$.
\quad \item If $\bar{u}_{Test} > \bar{u}_{min}$ then return $\bar{u}_{Test}$.
\item Else Loop
\end{enumerate}

\end{algorithm}

\section{A jump-diffusion stochastic volatility model on the asset dynamic}
We study a class of  multivariate affine models that contains large jumps and stochastic volatility  and   derive the characteristic function of their processes using the standard procedure involving Ito formula. The class can be view as a direct generalization to the multivariate case of Bates' model, see Bates(1996). Then we implement FFT approach on two particular models.\\
Specifically we consider a d-dimensional  model for the the movement of asset prices and volatilities as :
\begin{equation}\label{eq:model}\begin{split}
dS_t^{(m)} &= S_t^{(m)} \mu^{(m)} dt + S_t^{(m)} \sigma^{(m)} \sqrt{V_t^{(m)}} dW_{t}^{S(m)} + S_{{t}^{-}}^{(m)} d\widetilde{Z_t^{(m)}} \\
dV_t^{(m)} &= \xi^{(m)}(\eta^{(m)}-V_t^{(m)})dt + \theta^{(m)} \sqrt{V_t^{(m)}} dW_t^{V(m)}
\end{split}\end{equation}
for $m=1,2, \ldots,d$, where $W_{t}^{S(m)}$ and $W_{t}^{V(m)}$ are Wiener processes driving the movement of the $m$-th asset and volatility respectively, and $\widetilde{Z_t}$ is a compound Poisson process driven by a Poisson Process $\widetilde{N_t}$ of jump intensity factor $\lambda$ and independent common jump sizes $Y_j$ distributed  log-normally, i.e. $\log Y_j \sim N(\overline{k},\Delta^2)$ for $j=1,2,\ldots,\widetilde{N_t}$.\\
We introduce a correlation between Brownian motions by:
\begin{eqnarray}
d[W^{S(m)}_t,W_t^{S(n)}] & = & \rh{S_{(m)}}{S_{(n)}} dt  \nonumber \\
d[W_t^{S(m)},W_t^{V(n)}] & = & \rh{s_{(m)}}{v_{(n)}} dt \nonumber \\
d[(W_t^{V(m)},W_t^{v(n)}] & = & \rh{v_{(m)}}{v_{(n)}} dt
\end{eqnarray}
for $n \neq m$, where $[A,B]$ denotes the quadratic covariation of $A,B$. Also we assume $\widetilde{Z_t^{(m)}}$ are independent of $W_t^{S(n)}$ and $W_t^{v(n)}$ for any  $n \neq m$.\\
Applying Ito lemma to $X_t^{(m)} = \log S_t^{(m)}$ produces:
\begin{equation}\label{eq:Model}
dX_{t}^{(m)} = (r-\lambda \overline{k^{(m)}}-\frac{1}{2} {\sigma^{(m)}}^2 V_{t}^{(m)})  dt + \sigma^{(m)} \sqrt{V_{t}^{(m)}} dW_{t}^{S(m)} + dZ_t^{(m)}
\end{equation}
for $m=1,2, \ldots,d$, where $Z_t^{(m)}$ is a compound Poisson process with distributes multivariate normal $Z_t \sim N(\overline{k},\Delta^2)$. The drift component $\mu^{(m)}$ is fixed under the risk-neutral measure to be $r-\lambda \overline{k^{(m)}}$.\\
Since the jump and the continuous components of our model are independent, the characteristic function of $X_t$ is the product of the characteristic functions of each component:
\begin{equation*}
\phi_{X_t}(u) = \phi_{X_{t}^c}(u) \phi_{Z_t}(u)
\end{equation*}
where $X_{t}^c$ is the continuous part of $X_t$.\\
As such we first consider only the continuous component of our model, where we define by $X_{t}^{(m,c)}$  its m-th component. We have then:
\begin{equation}\label{eq:contpart}
{dX_{t}^{(m,c)}} = (r-\lambda\overline{k^{(m)}}-\frac{1}{2} {\sigma^{(m)}}^2 V_{t}^{(m)})  dt + \sigma^{(m)} \sqrt{V_{t}^{(m)}} dW_{t}^{S(m)}
\end{equation}
with
\begin{equation*}
dX_{t}^{(m)} = {dX_{t}^{(m,c)}} + dZ_{t}^{(m)}
\end{equation*}
Define the function:
\begin{equation}\label{eq:marti}
f(x,v,t,u) = \E_{\mathbb{Q}} (e^{i u \cdot X_{T}^{c}} | X_{t}^{c}=x, V_t=v)
\end{equation}
By standard application of Feyman-Kac formula we have for the  model described in (\ref{eq:contpart}) the characteristic function of the continuous component of $X_t$ satisfies the following PDE:
\begin{equation}\label{eq:pdechar}\begin{split}
0 &= \pdt + \sumi{d} \left[ \pd{x}{m} \left(r-\lambda\overline{k^{(m)}} - \frac{1}{2} {\sigma^{(m)}}^2 v^{(m)}\right)
+ \pd{v}{m} \left(\xi^{(m)}(\eta^{(m)}-v^{(m)}) \right) \right] \\
& \quad +\frac{1}{2} \sumij{d} \left[ \pdm{x}{m}{x}{n} \rh{s_{(m)}}{s_{(n)}} \sigma^{(m)} \sigma^{(n)} \sqrt{v^{(m)}{v^{(n)}}} \rml
\lml + 2 \pdm{x}{m}{v}{n} \rh{s_{(m)}}{v_{(n)}} \sigma^{(m)} \theta^{(n)} \sqrt{v^{(m)}{v^{(n)}}}  \rml
\lml + \pdm{v}{m}{v}{n} \rh{v_{(m)}}{v_{(n)}}\sqrt{v^{(m)}{v^{(n)}}} \theta^{(m)} \theta^{(n)}  \right]
\end{split}\end{equation}
with terminal condition $f(x,v,T,u) = e^{i u \cdot x}$.\\
 In general equation (\ref{eq:pdechar}) is non-linear in the coefficients and a close-form solution is not available without making further simplifying assumptions.
We  consider two specific cases of our general model where an affine structure is present, and limit ourselves to two assets.\\
 In the first case, we assume no correlation between the asset prices in the continuous component (although we allow for correlation through the jumps to capture dependence) which we will refer to as the independent volatility case (in the sense that each asset has an independent volatility driving its continuous component which it is correlated with). In the second case we extend the models of Dempster and Hong(2000) where proportional volatilities are considered to include jumps, which we refer to as the proportional volatility or common volatility case.

For the independent volatility case, we make the following assumptions:
\begin{eqnarray}\label{eq:IVassum1}
&\rh{s_{(m)}}{s_{(n)}} &= 0 \qquad  \textrm{for } m \neq n, \textrm{and} 1 \textrm{ for } m = n \label{eq:IVassum2} \\
&\rh{v_{(m)}}{v_{(n)}} &= 0 \qquad  \textrm{for } m \neq n, \textrm{and}  1 \textrm{ for } m = n \label{eq:IVassum3} \\
&\rh{s_{(m)}}{v_{(n)}} &= 0 \qquad  \textrm{for } m \neq n
\end{eqnarray}
With these simplifications in mind, we can solve the PDE given in (\ref{eq:pdechar}).

\begin{theorem}\label{CFIV}
For the market model described in (\ref{eq:model}), with the assumptions given in (\ref{eq:IVassum1}-\ref{eq:IVassum3}) the characteristic function of $X_t$, $0 \leq t \leq T$ is given by:
\begin{eqnarray*}\label{eq:charIV} \nonumber
\phi_{X_t}(u) &=& \phi_{X_{t}^c}(u) \phi_{Z_t}(u) \\
\textrm{where} \nonumber \\
\phi_{X_{t}^c}(u) &=& e^{i u \cdot X_0 + C(T-t) + V_0 \cdot D(T-t)} \nonumber \\
\phi_{Z_t}(u) &=& e^{t \lambda (exp(i u^T \overline{k} - \frac{1}{2} u^T \Delta u) - 1)} \nonumber \\
D_{(m)}(s) &=& \frac{2 \zeta^{(m)} (1 - e^{-\gamma^{(m)} s})}{2 \gamma^{(m)} - (\gamma^{(m)} - \omega^{(m)}) (1- e^{-\gamma^{(m)} s})} \nonumber \\
C(s) &=&  \sumi{2} \Big( i u^{(m)} \big(r-\lambda \overline{k^{(m)}}\big) \Big) s \nonumber \\
 &-& \frac{\xi^{(m)} \eta^{(m)} }{{\theta^{(m)}}^2} \left[2 \ln \left(\frac{2 \gamma^{(m)} - (\gamma^{(m)} - \omega^{(m)})(1- e^{-\gamma^{(m)} s})}{2 \gamma^{(m)}} \right) + (\gamma^{(m)} - \omega^{(m)}) s \right] \quad  \nonumber \\
\zeta^{(m)} &=&  -\frac{1}{2} {\theta^{(m)}}^2 (i u^{(m)} \sigma^{(m)} + {u^{(m)}}^2 {\sigma^{(m)}}^2 )  \nonumber \\
\omega^{(m)} &=&  \xi^{(m)} - i \theta^{(m)} \sigma^{(m)} \rh{s_{(m)}}{v_{(m)}} u^{(m)} \nonumber \\
\gamma^{(m)} &=& \sqrt{ {\omega^{(m)}}^2 - 2 {\theta^{(m)}}^2 \zeta^{(m)}} \nonumber \\
\end{eqnarray*}
where $D_{(m)},\zeta^{(m)},\omega^{(m)}$ and $\gamma^{(m)}$ are respectively the m-th components of vectors $D, \zeta, \omega$ and $\gamma$.
\end{theorem}
\begin{proof} Our assumptions reduce the equation given in (\ref{eq:pdechar}) to:
\begin{equation*}\label{pdeindvol}\begin{split}
0 &= \sumi{2} \left[ \pd{x}{m} \left(r-\lambda \overline{k^{(m)}}-\frac{1}{2} {\sigma^{(m)}}^2 \vi \right)
+ \pd{v}{m}\left(\xi^{(m)}(\eta^{(m)}-v^{(m)}) \right) \rml
\lml + \frac{1}{2} \left(\pds{x}{m} {\sigma^{(m)}}^2 \vi + 2\pdm{x}{m}{v}{m} \rh{s_{(m)}}{v_{(m)}}  \sigma^{(m)} \vi \theta^{(m)} + \pds{v}{m} \vi {\theta^{(m)}}^2  \right) \right] + \pdt
\end{split}\end{equation*}
We now guess a solution of the form:
\begin{equation*}
f(x,v,t,u) = e^{i u \cdot x + C(T-t) + v \cdot D(T-t)}
\end{equation*}
where $C(T-t)$ and
\begin{equation*}
D(T-t) =
\left[ \begin{array}{c}
D_1 (T-t) \\
D_2 (T-t) \\
\end{array} \right]
\end{equation*}
are functions of $t$ alone. Applying our guess it gives us:
\begin{equation*}\label{INSERT}\begin{split}
\frac{dC}{dt} + \sumi{2} \frac{dD_{(m)}}{dt} \vi &= \sumi{2} \left[ \bigg(iu^{(m)} \Big(r-\lambda \overline{k^{(m)}}-\frac{1}{2} {\sigma^{(m)}}^2 \vi \Big)\bigg) \rml
\lml + D_{(m)} \xi^{(m)}(\eta^{(m)}-v^{(m)}) + \frac{1}{2} \left(-{u^{(m)}}^2 {\sigma^{(m)}}^2 \vi \right. \rml
\lml \left. + D_{(m)}^2 \vi {\theta^{(m)}}^2 + 2(i u^{(m)} D_{(m)} \rh{s_{(m)}}{v_{(m)}} \sigma^{(m)} \vi \theta^{(m)}) \right) \right]
\end{split}\end{equation*}
 This produces a series of Riccati ODE's for $m=1,2$:
\begin{eqnarray*}\label{ode}
 \frac{dD_{(m)}}{dt}&=&  -\frac{1}{2} {\sigma^{(m)}}^2 \left(i u^{(m)}  + {u^{(m)}}^2\right) - \left(\xi^{(m)} - i u^{(m)} \rh{s_{(m)}}{v_{(m)}} \sigma^{(m)} \theta^{(m)}\right)D_{(m)} \nonumber\\
&+& \frac{1}{2}{\theta^{(m)}}^2 {D_{(m)}}^2 \\
\textrm{for } m &=& 1,2 \textrm{, and} \nonumber\\
\frac{dC}{dt} &=& \sumi{2}\bigg(\Big(i u^{(m)} \big(r-\lambda \overline{k^{(m)}}\big)\Big) + \xi^{(m)} \eta^{(m)} D_{(m)}(s)\bigg)
\end{eqnarray*}
which, following from the terminal condition on equation (\ref{eq:pdechar}), have initial conditions $D_{(m)}(0) = 0$ and $C(0) = 0$. The solutions to these equations are given in Theorem 3.2 above. It should be noted that because the continuous components of the asset prices are independent of each other, the characteristic function given above can also be written as:
\begin{equation*}
\phi_{X_t}(u) = \prod_{m=1}^d \phi_{X_{t}^{(m,c)}}(u)  \phi_{Z_t}(u)
\end{equation*}
where $\phi_{X_{t}^{(m),c}}(u)$ is the characteristic function of the continuous component of the $m^{th}$ asset, as given in Albrecher et al.(2007).\\
Looking at the jump component, we know from the L\'{e}vy–Khinchine formula that its characteristic function is given by:
\begin{equation*}
\phi_{Z_t}(u) = e^{t \lambda (exp(i u' \overline{k} - \frac{1}{2} u' \Delta u) - 1)}
\end{equation*}
where $'$ indicates the transpose operator.
\end{proof}
We now consider the case of Proportional Stochastic Volatilities. In this situation we require only one volatility process for all of the assets and use the parameter $\sigma^{(m)}$ to allow for varying volatilities between the assets. Our model for the log-prices and volatility is thus:
\begin{eqnarray} \label{eq:commonvols}
dX_{t}^{(m)} &=& {dX_{t}^{(m,c)}} + dZ_{t}^{(m)} \nonumber \\
{dX_{t}^{(m,c)}} &=& (r - \lambda \overline{k^{(m)}}-  \frac{1}{2} {\sigma^{(m)}}^2 V_{t})  dt + \sigma^{(m)} \sqrt{V_{t}} dW_{t}^{S(m)} \nonumber \\
dV_t &=& \xi(\eta-V_t)dt + \theta \sqrt{V_t} dW_t^{V}
\end{eqnarray}
By a similar procedure than in Theorem \ref{CFIV} we have:
\begin{theorem}\label{CFIV}
For the market model described in (\ref{eq:commonvols}), the characteristic function of $X_t$ is given by:
\begin{eqnarray*}
\phi_{X_t}(u) &=& \phi_{X_{t}^c}(u) \phi_{Z_t}(u) \\
\textrm{where} \nonumber \\
\phi_{X_{t}^c}(u) &=& e^{i u \cdot X_0 + C(T-t) + V_0 D(T-t)}  \\
D(s) &=& \frac{2 \zeta (1 - e^{-\gamma s})}{2 \gamma - (\gamma - \omega) (1- e^{-\gamma s})} \nonumber \\
C(s) &=& \Big( \sumi{d} i u^{(m)} \big(r-\lambda \overline{k^{(m)}}\big) \Big) s - \frac{\xi \eta }{\theta^2} \left[2 \ln \left(\frac{2 \gamma - (\gamma - \omega)(1- e^{-\gamma s})}{2 \gamma} \right) + (\gamma - \omega) s \right] \quad  \nonumber \\
\zeta &=&  -\frac{1}{2} \Big[ \sumi{d} i {\sigma^{(m)}}^2 u^{(m)} + \sumij{d} \sigma^{(m)} \sigma^{(n)} u^{(m)} u^{(n)} \rh{s_{(m)}}{s_{(n)}} \Big] \nonumber \\
\omega &=&  \xi - i \theta (\sumi{2} \rh{s_{(m)}}{v} u^{(m)}) \nonumber \\
\gamma &=& \sqrt{ \omega^2 - 2 \theta^2 \zeta} \nonumber \\
\end{eqnarray*}
with $\phi_{Z_t}(u)$ defined as in Theorem 3.2 above.
\end{theorem}

\section{Numerical Computation of Spread Option Prices}
We present our numerical results. In the first part of this section we show the analyze  the implementation of algorithm described in (\ref{alg:grid}) to price spreads under both models considered in the previous section, while in the second part we provide a sensitivity analysis of spread prices with respect to some parameters in our models.
\subsection{A comparison between FFT and Monte Carlo prices}
Table \ref{tbl:errs} compares the results produced by Monte Carlo simulation with the results from a FFT method applied to the proportional volatility case with various strikes as given in its first column. We use a maturity equal to one year. The second column gives the prices obtained from one million Monte Carlo simulations of 2000 time steps each, while the subsequent columns present the relative error, measured in percent,  obtained using the FFT method for various values of the discretization points in the grid. We use the following benchmark set of parameters:\\
$S_0^{(1)} = 100$ , $S_0^{(2)} = 96$, $\sigma^{(1)} = 1$, $\sigma^{(2)} = 0.5$, $\xi = 1$, $\eta = 0.04$, $\theta = 0.05$, $V_0 = 0.04$, $\lambda = 1$, $\overline{k^{(1)}} = \overline{k^{(2)}} = 0.05$, $\delta^{(1)} = \delta^{(2)} = 0.05$, $\rh{S(1)}{S(2)} = 0.5$, $\rh{S(1)}{V} = -0.5$, $\rh{S(2)}{V} = 0.25$, $r_f = 0.1$.

\begin{table}[htb!]
\caption{Comparison of Monte Carlo method with FFT method for proportional volatility case for $\overline{u}_{min}$ = 20 under the benchmark set of parameters.}

\label{tbl:errs}
\begin{tabular}{|l|l|l|l|l|l|}
\hline
\textbf{K} & \textbf{MC} & \textbf{128} & \textbf{256} & \textbf{512} & \textbf{1024} \\ \hline
\textbf{2} & 8.359781 & -0.008902 & -0.008902 & -0.008855 & -0.008902 \\ \hline
\textbf{2.2} & 8.264856 & -0.009001 & -0.009001 & -0.008967 & -0.009 \\ \hline
\textbf{2.4} & 8.170669 & -0.009127 & -0.009127 & -0.009081 & -0.009127 \\ \hline
\textbf{2.6} & 8.063694 & -0.007571 & -0.00757 & -0.007533 & -0.00757 \\ \hline
\textbf{2.8} & 7.984489 & -0.009357 & -0.009356 & -0.009308 & -0.009356 \\ \hline
\textbf{3} & 7.879148 & -0.007781 & -0.00778 & -0.007743 & -0.00778 \\ \hline
\textbf{3.2} & 7.787975 & -0.007905 & -0.007905 & -0.007849 & -0.007905 \\ \hline
\textbf{3.4} & 7.697545 & -0.008 & -0.008 & -0.007956 & -0.008 \\ \hline
\textbf{3.6} & 7.633466 & -0.011431 & -0.011431 & -0.011392 & -0.011431 \\ \hline
\textbf{3.8} & 7.544433 & -0.011586 & -0.011586 & -0.011531 & -0.011586 \\ \hline
\textbf{4} & 7.456122 & -0.010902 & -0.010902 & -0.010856 & -0.010901 \\ \hline
\end{tabular}
\end{table}
As we observe the FFT method provides an accurate value for the price of a spread option under the benchmark parameter set considered, although consistently biased low of Monte Carlo. Similar results regarding speed and accuracy are obtained for the model of independent volatilities.\\
In both cases it is interesting to note that an increasing in the number of points taken in the grid, after 128 points, does not result in the corresponding error reduction.\\
In addition to the accuracy of the method, it is also useful to compare the computational effort required to price spread options under each method. We can compare the execution times for the FFT method for both the proportional volatility case and the independent volatility case, as well as for  Monte Carlo simulations, which are given in table \ref{tbl:times}.
\begin{table}[htb!]
\begin{center}
  \caption{Run times, in seconds, for FFT vs. Monte Carlo Simulation, $\bar{u}_{min} = 20$ and under the benchmark set of parameters.}
  \label{tbl:times}
\begin{tabular}{|l|l|l|}
  \hline
   \textbf{Grid Size} &  \textbf{proportional volatility} & \textbf{Independent Volatilities} \\
   \hline
  64 &  0.020426  & 0.054677  \\
  \hline
  128 &  0.050574 & 0.084188  \\
  \hline
  256 & 0.233219   & 0.346391   \\
  \hline
  512 &  0.997439 & 1.454524   \\
  \hline
  1024 &  3.970557 &  5.920600 \\
  \hline
  MC &  1368.67 & 1503.42   \\
  \hline
  \end{tabular}
  \end{center}
\end{table}
This is where the benefits of the FFT method become more clear, as it largely outperforms the Monte Carlo method. Moreover, because of the additional volatility introduced when we add a jump component to our model, the Monte Carlo method is very slow to converge and requires a high number of simulations and a fine grid to accurately generate a price. The proportional volatility Case also outperforms the independent volatility one, due mainly to the fact that  the later case essentially requires the calculation of three characteristic functions (one for each asset-volatility pair, and one for the correlated jumps), while our proportional volatility case requires only two (one for the continuous component, and one for the jumps).
\\
Finally, we compare the FFT prices produced by the proportional volatility and independent volatility Cases. Fixing all the parameters but those leading to the correlation, since this is essentially the area where the two models differ, Figure \ref{ch:IndVsProp} shows the difference in prices by the magnitude $P_{prop} - P_{ind}$, i.e. the difference between the price in the proportional volatility and the independent cases, for pairs of  correlation between assets. Note that while the independent case does not have correlation between the Wiener processes driving the continuous component, it does have correlation in the jump components. The results appear as we would expect them to, with the independent volatility case producing a higher prices when the correlation in the proportional volatility case is high, while the reverse also holds when the correlation in the proportional volatility case is low. It should be noted that we see a much greater difference in prices when we decrease the Common Volatility correlation, showing the price in the Common Volatility model is much more dependent on the correlation. Furthermore the latter changes in a more pronounced non-linear manner.

\begin{figure}[htb!]
    \begin{center}
        \includegraphics[width=\textwidth]{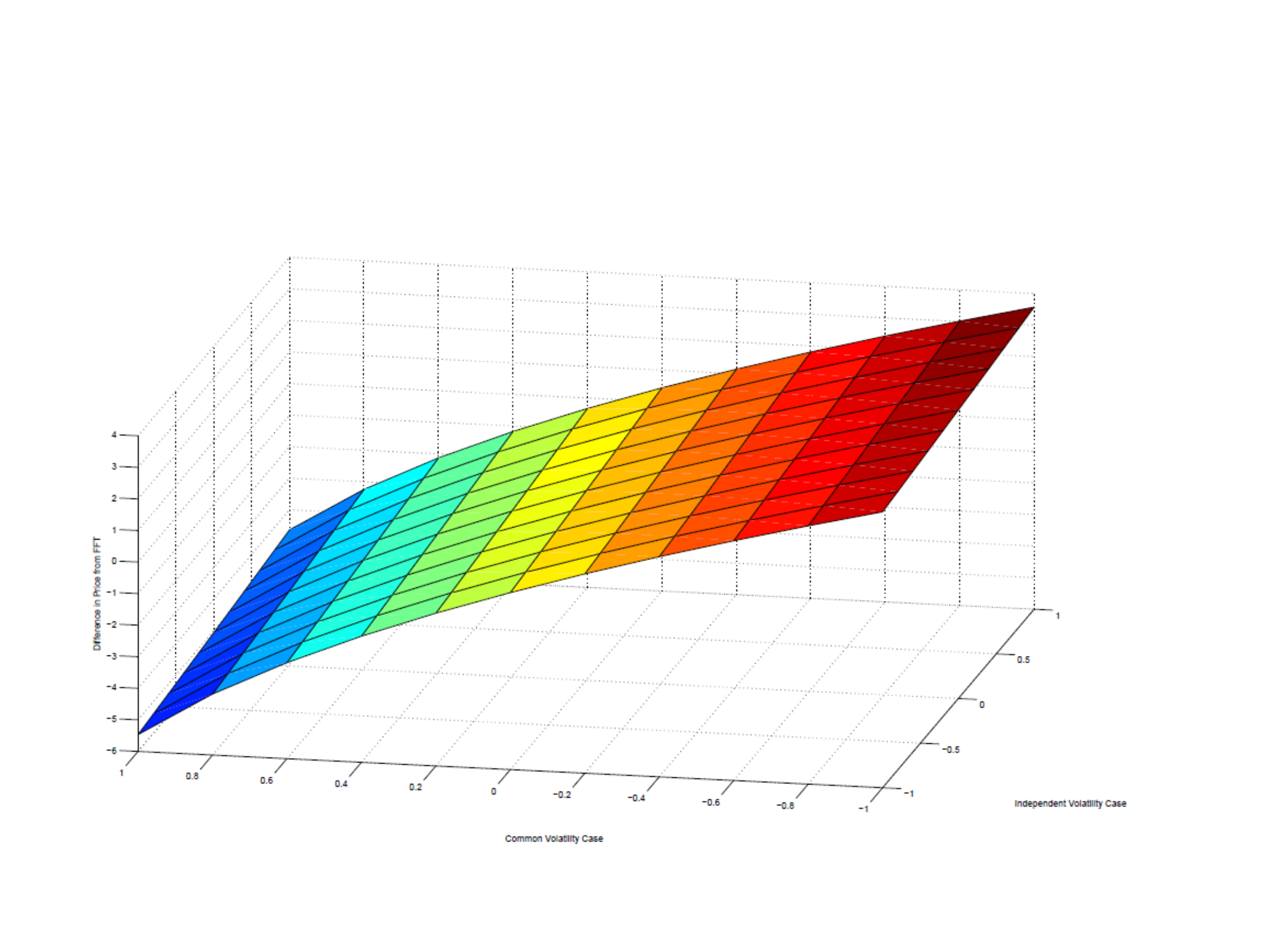}
        \caption{Price difference between proportional volatility model and independent volatility model with asset correlation $\rh{s(1)}{s(2)}$}
        \label{ch:IndVsProp}
    \end{center}
\end{figure}

\subsection{Parameter Sensitivities}

We now consider just the proportional volatility case, and examine the sensitivities of the price to the various parameters in our model.  We fix the benchmark set of parameters as in the previous subsection, additionally we select a grid size of $N = 512$, damping parameter $\epsilon = (-3,1)$, and a minimum truncation interval of $\overline{u}_{min} = 40$, with the actual truncation interval selected using Algorithm (\ref{alg:grid}).
\begin{figure}[htb!]
    \begin{center}
        \includegraphics[width=\textwidth]{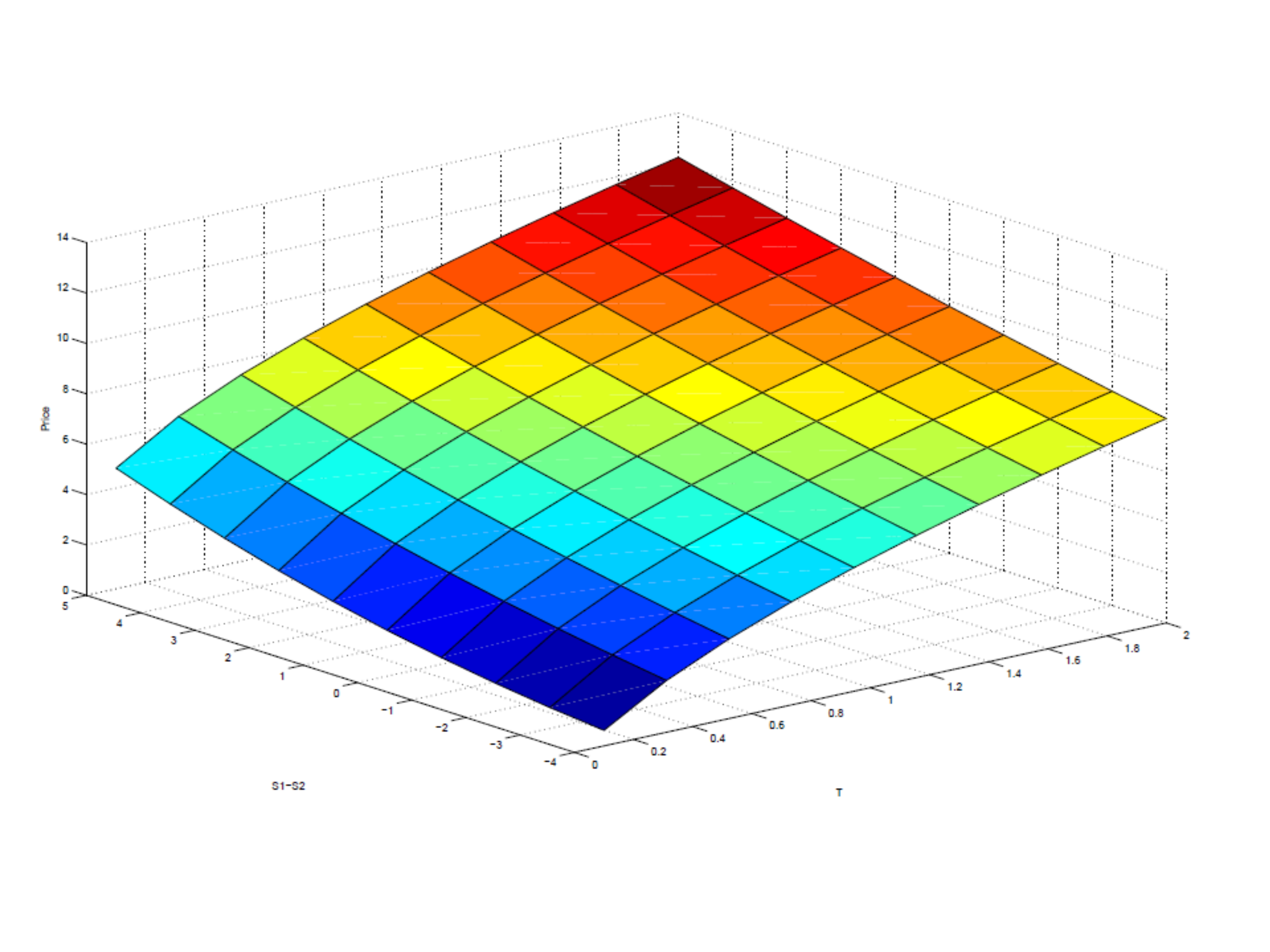}
                \caption{Price with Variation in Moneyness and Time to Maturity under the benchmark set of parameters}

        \label{ch:moneymaturity}
    \end{center}
\end{figure}

Figure \ref{ch:moneymaturity} and Table \ref{tbl:moneyT} show the price as we vary both the initial moneyness of our option and the time to maturity. Obviously we would expect the price of the option to increase as the option ranges from out-of-the-money to in the money, and also to increase as the time to maturity rises, both facts which we observe. The data shows that a large variation based on the moneyness and time to maturity, and the plot also shows that at shorter maturities the effect of increasing the moneyness of the option is much greater, which we would again expect.
\begin{table}[htb!]
\caption{Comparison of Prices for case for Variation in Moneyness and Time to Maturity $T$}
\label{tbl:moneyT}
\begin{tabular}{|r|l|l|l|l|l|l|l|l|l|l|}
\hline
\textbf{T} & \textbf{0.1} & \textbf{0.31} & \textbf{0.52} & \textbf{0.73} & \textbf{0.94} & \textbf{1.16} & \textbf{1.37} & \textbf{1.58} & \textbf{1.79} & \textbf{2} \\
$S_0^{(1)}-S_0^{(2)}$ &&&&&&&&&&\\ \hline
\textbf{5} & 4.79 & 6.3 & 7.41 & 8.33 & 9.13 & 9.86 & 10.52 & 11.13 & 11.71 & 12.25 \\ \hline
\textbf{4} & 4.06 & 5.67 & 6.82 & 7.76 & 8.57 & 9.3 & 9.97 & 10.6 & 11.18 & 11.73 \\ \hline
\textbf{3} & 3.4 & 5.08 & 6.25 & 7.21 & 8.03 & 8.78 & 9.45 & 10.08 & 10.67 & 11.22 \\ \hline
\textbf{2} & 2.8 & 4.53 & 5.72 & 6.69 & 7.52 & 8.27 & 8.95 & 9.58 & 10.17 & 10.73 \\ \hline
\textbf{1} & 2.28 & 4.02 & 5.22 & 6.19 & 7.03 & 7.78 & 8.46 & 9.1 & 9.69 & 10.26 \\ \hline
\textbf{0} & 1.82 & 3.55 & 4.75 & 5.72 & 6.56 & 7.31 & 8 & 8.64 & 9.23 & 9.8 \\ \hline
\textbf{-1} & 1.43 & 3.12 & 4.31 & 5.27 & 6.11 & 6.87 & 7.55 & 8.19 & 8.79 & 9.35 \\ \hline
\textbf{-2} & 1.1 & 2.73 & 3.9 & 4.86 & 5.69 & 6.44 & 7.13 & 7.76 & 8.36 & 8.92 \\ \hline
\textbf{-3} & 0.83 & 2.37 & 3.51 & 4.46 & 5.29 & 6.03 & 6.72 & 7.35 & 7.95 & 8.51 \\ \hline
\textbf{-4} & 0.62 & 2.05 & 3.16 & 4.09 & 4.91 & 5.65 & 6.32 & 6.96 & 7.55 & 8.11 \\ \hline
\end{tabular}
\end{table}

\begin{figure}[htb!]
    \begin{center}
        \includegraphics[width=\textwidth]{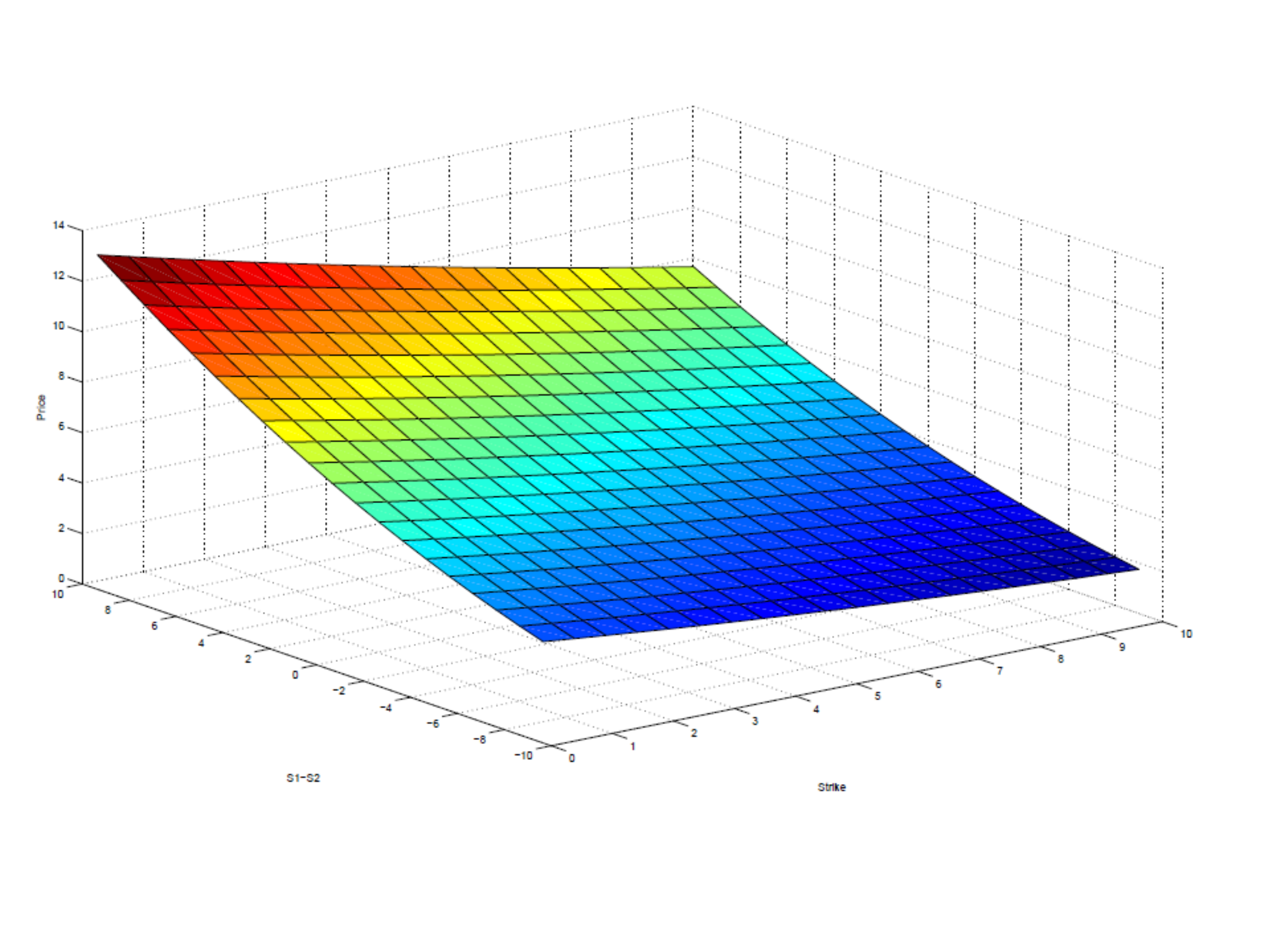}
        \caption{Price with Variation in Moneyness and Strike, the remaining parameters belong to the benchmark setting}

        \label{ch:moneystrike}
    \end{center}
\end{figure}

In Figure \ref{ch:moneystrike} we look at how moneyness and strike affect the price of the our spread option. We again observe a price increasing as both the moneyness increases, and as the strike decreases.
\begin{table}[htb!]
\caption{Comparison of Prices for case for Variation in Moneyness and Strike $K$,  the remaining parameters belong to the benchmark setting}
\label{tbl:moneyK}
\begin{tabular}{|l|l|l|l|l|l|l|l|l|l|l|}
\hline
\textbf{$S^{(1)}-S^{(2)}$/$K$} & \textbf{0.25} & \textbf{1.28} & \textbf{2.3} & \textbf{3.33} & \textbf{4.36} & \textbf{5.38} & \textbf{6.41} & \textbf{7.43} & \textbf{8.46} & \textbf{9.49} \\ \hline
\textbf{5} & 12.92 & 12.28 & 11.67 & 11.07 & 10.49 & 9.93 & 9.39 & 8.87 & 8.37 & 7.89 \\ \hline
\textbf{4} & 11.57 & 10.98 & 10.4 & 9.84 & 9.31 & 8.79 & 8.29 & 7.81 & 7.35 & 6.91 \\ \hline
\textbf{3} & 10.32 & 9.76 & 9.23 & 8.71 & 8.22 & 7.74 & 7.28 & 6.85 & 6.43 & 6.03 \\ \hline
\textbf{2} & 9.15 & 8.64 & 8.14 & 7.67 & 7.21 & 6.78 & 6.36 & 5.97 & 5.59 & 5.23 \\ \hline
\textbf{1} & 8.07 & 7.6 & 7.15 & 6.71 & 6.3 & 5.91 & 5.53 & 5.17 & 4.83 & 4.51 \\ \hline
\textbf{0} & 7.08 & 6.65 & 6.24 & 5.85 & 5.47 & 5.12 & 4.78 & 4.46 & 4.16 & 3.87 \\ \hline
\textbf{-1} & 6.18 & 5.79 & 5.42 & 5.07 & 4.73 & 4.41 & 4.11 & 3.83 & 3.56 & 3.31 \\ \hline
\textbf{-2} & 5.37 & 5.02 & 4.68 & 4.37 & 4.07 & 3.79 & 3.52 & 3.27 & 3.03 & 2.81 \\ \hline
\textbf{-3} & 4.64 & 4.32 & 4.02 & 3.74 & 3.48 & 3.23 & 3 & 2.77 & 2.57 & 2.37 \\ \hline
\textbf{-4} & 3.98 & 3.7 & 3.44 & 3.19 & 2.96 & 2.74 & 2.54 & 2.34 & 2.16 & 2 \\ \hline
\end{tabular}
\end{table}

We also notice that for a given strike the change in price is non-linear as a function of the moneyness, an indication of a skew in the implied volatilities. While the calculation of the implied volatilities and correlations is beyond the scope of this paper, it should be noted that the skews and smiles in univariate models tend to occur because the market is overestimating the volatility of away from the money options. Thus the non-linearity we see in the price can be taken as evidence of the presence of a skew in our implied volatilities and correlations, as we would expect with both a stochastic volatility and jump component.
\begin{figure}[htb!]
    \begin{center}
        \includegraphics[width=\textwidth]{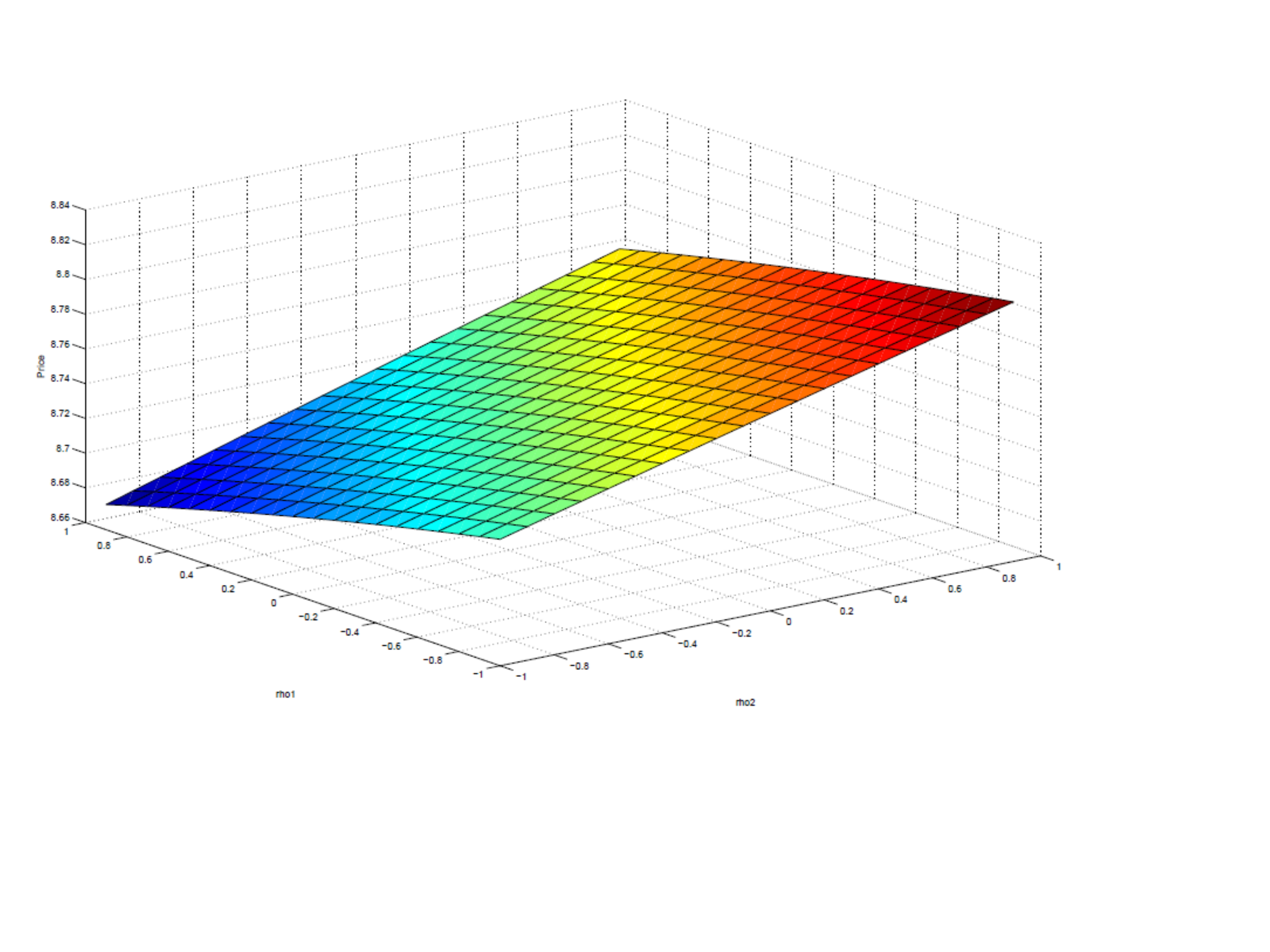}
         \caption{Price with Variation in Asset-Volatility Correlation $\rh{s(m)}{v}$, $\rh{s(1)}{s(2)} = +0.5$,  the remaining parameters belong to the benchmark setting}

        \label{ch:corrvol}
    \end{center}
\end{figure}

We consider next the variation of the correlation between each asset and the driving volatility process. We look first at the case where the assets have positive correlation, as given in Figure \ref{ch:corrvol}. In this plot we see an increase in the price as the correlation between the short asset and the volatility increases, while we see also observe an increase in the price for a constant $\rh{s(2)}{v}$ as we decrease the value of $\rh{s(1)}{v}$. The highest price is observed for a strong correlation between the short asset and the volatility, and a strong negative correlation between the the long asset and the volatility.
\begin{figure}[htb!]
    \begin{center}
        \includegraphics[width=\textwidth]{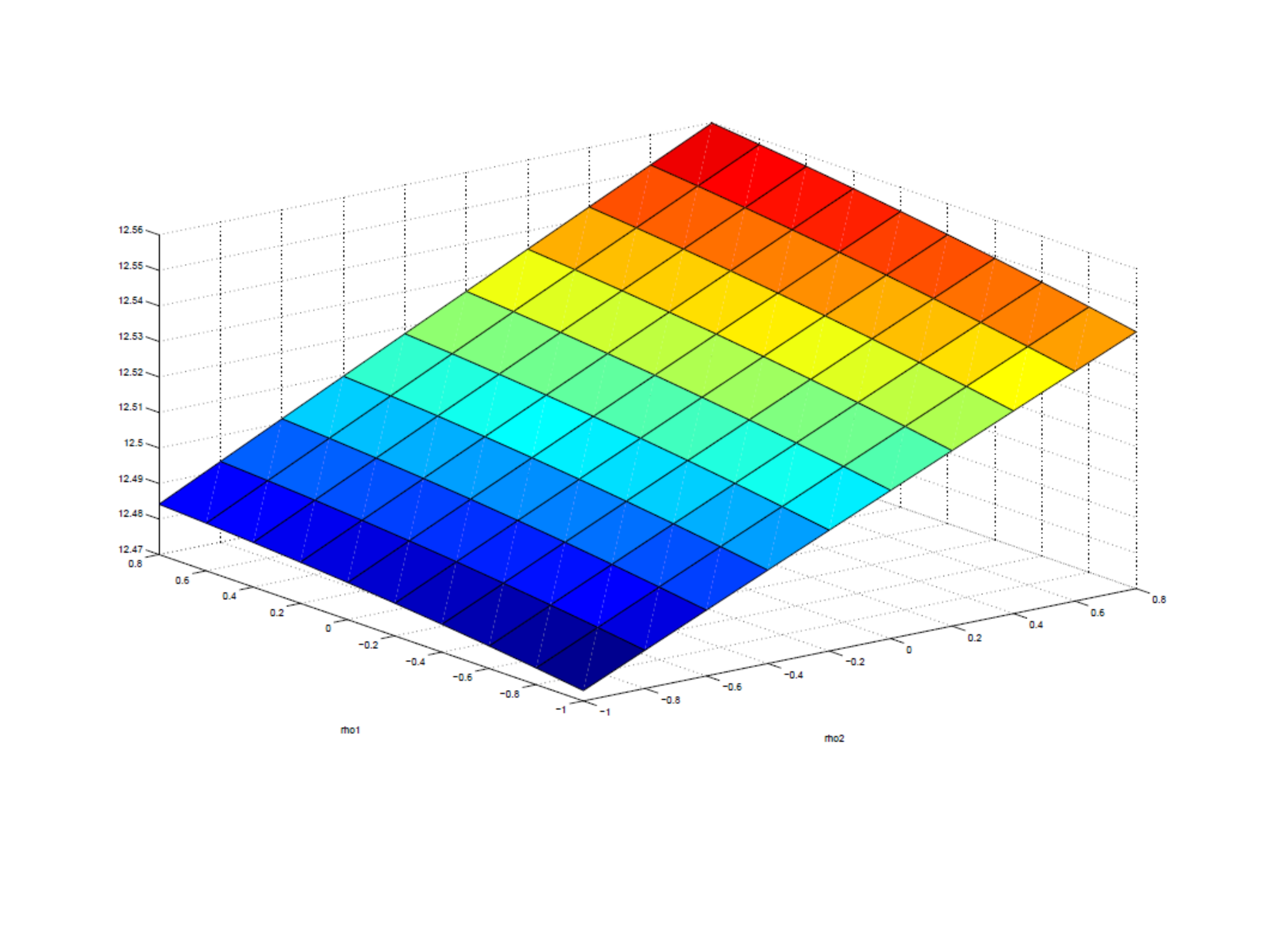}
         \caption{Price with Variation in Asset-Volatility Correlation $\rh{s(m)}{v}$, $\rh{s(1)}{s(2)} = -0.5$,  the remaining parameters belong to the benchmark setting}

        \label{ch:corrvolneg}
    \end{center}
\end{figure}

In the alternative case, as shown in Figure \ref{ch:corrvolneg}, we see the surface rotated, with the maximum price occurring when both $\rh{s(1)}{v}$ and $\rh{s(2)}{v} = 1$. For both cases however, as Table \ref{tbl:corrVol} shows, the selection of correlation parameters has very small effects on the price produced, as we see little variation between the high and low prices produced in this example.
\begin{table}[htb!]
\caption{Comparison of Prices for Variation in Asset-Volatility Correlation $\rh{s(m)}{v}$ with $\rh{s(1)}{s(2)} = -0.5$}
\label{tbl:corrVol}
\begin{tabular}{|l|l|l|l|l|l|l|l|l|l|l|}
\hline
\textbf{$\rh{s(1)}{v}$/$\rh{s(2)}{v}$} & \textbf{-1} & \textbf{-0.8} & \textbf{-0.6} & \textbf{-0.4} & \textbf{-0.2} & \textbf{0} & \textbf{0.2} & \textbf{0.4} & \textbf{0.6} & \textbf{0.8} \\ \hline
\textbf{-1} & 12.47 & 12.48 & 12.49 & 12.5 & 12.5 & 12.51 & 12.52 & 12.53 & 12.53 & 12.54 \\ \hline
\textbf{-0.8} & 12.47 & 12.48 & 12.49 & 12.5 & 12.51 & 12.51 & 12.52 & 12.53 & 12.54 & 12.54 \\ \hline
\textbf{-0.6} & 12.48 & 12.48 & 12.49 & 12.5 & 12.51 & 12.52 & 12.52 & 12.53 & 12.54 & 12.55 \\ \hline
\textbf{-0.4} & 12.48 & 12.49 & 12.49 & 12.5 & 12.51 & 12.52 & 12.53 & 12.53 & 12.54 & 12.55 \\ \hline
\textbf{-0.2} & 12.48 & 12.49 & 12.5 & 12.5 & 12.51 & 12.52 & 12.53 & 12.54 & 12.54 & 12.55 \\ \hline
\textbf{0} & 12.48 & 12.49 & 12.5 & 12.51 & 12.51 & 12.52 & 12.53 & 12.54 & 12.55 & 12.55 \\ \hline
\textbf{0.2} & 12.48 & 12.49 & 12.5 & 12.51 & 12.51 & 12.52 & 12.53 & 12.54 & 12.55 & 12.56 \\ \hline
\textbf{0.4} & 12.48 & 12.49 & 12.5 & 12.51 & 12.52 & 12.52 & 12.53 & 12.54 & 12.55 & 12.56 \\ \hline
\textbf{0.6} & 12.48 & 12.49 & 12.5 & 12.51 & 12.52 & 12.53 & 12.53 & 12.54 & 12.55 & 12.56 \\ \hline
\textbf{0.8} & 12.48 & 12.49 & 12.5 & 12.51 & 12.52 & 12.53 & 12.53 & 12.54 & 12.55 & 12.56 \\ \hline
\end{tabular}
\end{table}

Looking at the effects of the jump frequency and asset correlation in Figure \ref{ch:JumpCorr} we see interesting results. While we see the expected results of the price increasing with both the frequency of jumps, and as our correlation moves towards $-1$, the increase in price is not linear.
\begin{figure}[htb!]
    \begin{center}
        \includegraphics[width=\textwidth]{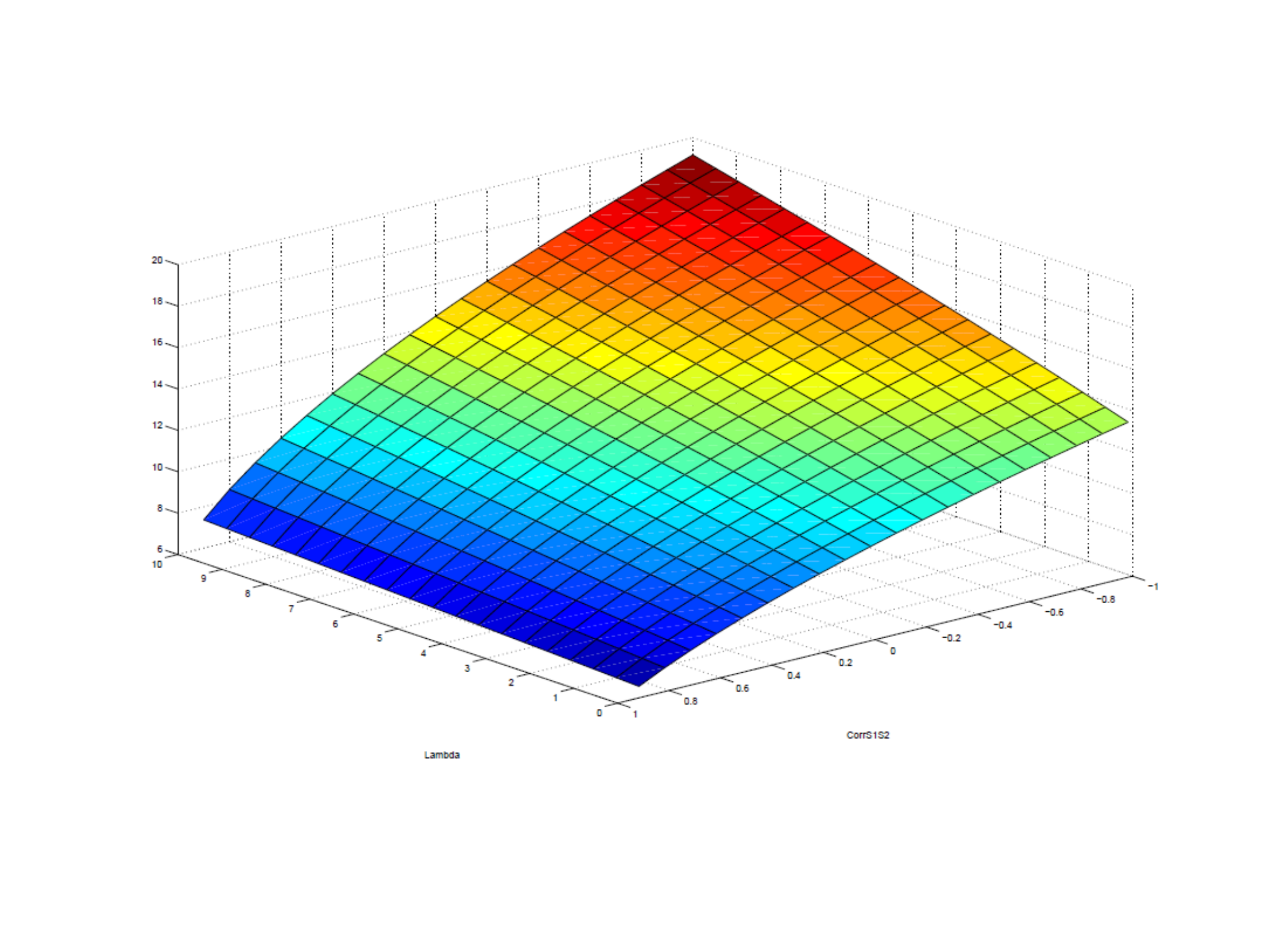}
        \caption{Price with Variation in Jump Frequency and Asset Correlation,  the remaining parameters belong to the benchmark setting}

        \label{ch:JumpCorr}
    \end{center}
\end{figure}

 As seen in table \ref{tbl:LambdaCorr} for lower jump frequencies the effects of decreasing the correlation tends to decrease in the area around -1, while for higher jump frequencies the effect remains strong right to $\rh{s(1)}{s(2)} = -1$. Obviously as we increase the jump frequency for negative correlations we expect more jumps to occur, and as such we expect more sudden movements of the underlying asset prices in opposite directions, which manifests itself in the higher prices produced under our model. We also observe that both parameters have a strong effect on the price, as we see it range from the low end of 7.02 dollars up to 18.90 dollars at the high end.
\begin{table}[htb!]
\caption{Comparison of Prices for Variation in Jump Frequency $\lambda$ and Asset Correlation $\rh{s(1)}{s(2)}$}
\label{tbl:LambdaCorr}
\begin{tabular}{|l|l|l|l|l|l|l|l|l|l|l|}
\hline
\textbf{$\lambda$/$\rh{s(1)}{s(2)}$} & \textbf{-1} & \textbf{-0.8} & \textbf{-0.6} & \textbf{-0.4} & \textbf{-0.2} & \textbf{0} & \textbf{0.2} & \textbf{0.4} & \textbf{0.6} & \textbf{0.8} \\ \hline
\textbf{0.1} & 13.37 & 14.07 & 14.74 & 15.39 & 16.02 & 16.63 & 17.22 & 17.79 & 18.35 & 18.9 \\ \hline
\textbf{1.14} & 12.85 & 13.51 & 14.14 & 14.76 & 15.35 & 15.92 & 16.48 & 17.03 & 17.56 & 18.08 \\ \hline
\textbf{2.18} & 12.3 & 12.92 & 13.51 & 14.09 & 14.64 & 15.18 & 15.71 & 16.22 & 16.72 & 17.21 \\ \hline
\textbf{3.23} & 11.71 & 12.29 & 12.84 & 13.38 & 13.9 & 14.4 & 14.89 & 15.37 & 15.84 & 16.29 \\ \hline
\textbf{4.27} & 11.1 & 11.62 & 12.13 & 12.62 & 13.1 & 13.56 & 14.01 & 14.46 & 14.89 & 15.31 \\ \hline
\textbf{5.31} & 10.43 & 10.91 & 11.37 & 11.81 & 12.24 & 12.66 & 13.07 & 13.47 & 13.86 & 14.24 \\ \hline
\textbf{6.35} & 9.72 & 10.13 & 10.54 & 10.93 & 11.31 & 11.68 & 12.04 & 12.39 & 12.74 & 13.08 \\ \hline
\textbf{7.39} & 8.93 & 9.28 & 9.62 & 9.95 & 10.27 & 10.58 & 10.89 & 11.19 & 11.49 & 11.77 \\ \hline
\textbf{8.44} & 8.05 & 8.32 & 8.58 & 8.84 & 9.09 & 9.33 & 9.57 & 9.81 & 10.04 & 10.27 \\ \hline
\textbf{9.48} & 7.02 & 7.19 & 7.35 & 7.52 & 7.68 & 7.83 & 7.99 & 8.14 & 8.3 & 8.45 \\ \hline
\end{tabular}
\end{table}
It is also interesting to observe the effects of varying the mean jump size parameter, $\overline{k}$, for each asset. As Figure \ref{ch:JumpSize} shows we see a large increase in the price produced as we move away from equal expected jump sizes. While this result is expected (if we believe one asset will jump with larger magnitude than the other than naturally we should believe that the spread between the assets will change), the extent of the variation is interesting to observe.
\begin{figure}[htb!]
    \begin{center}
        \includegraphics[width=\textwidth]{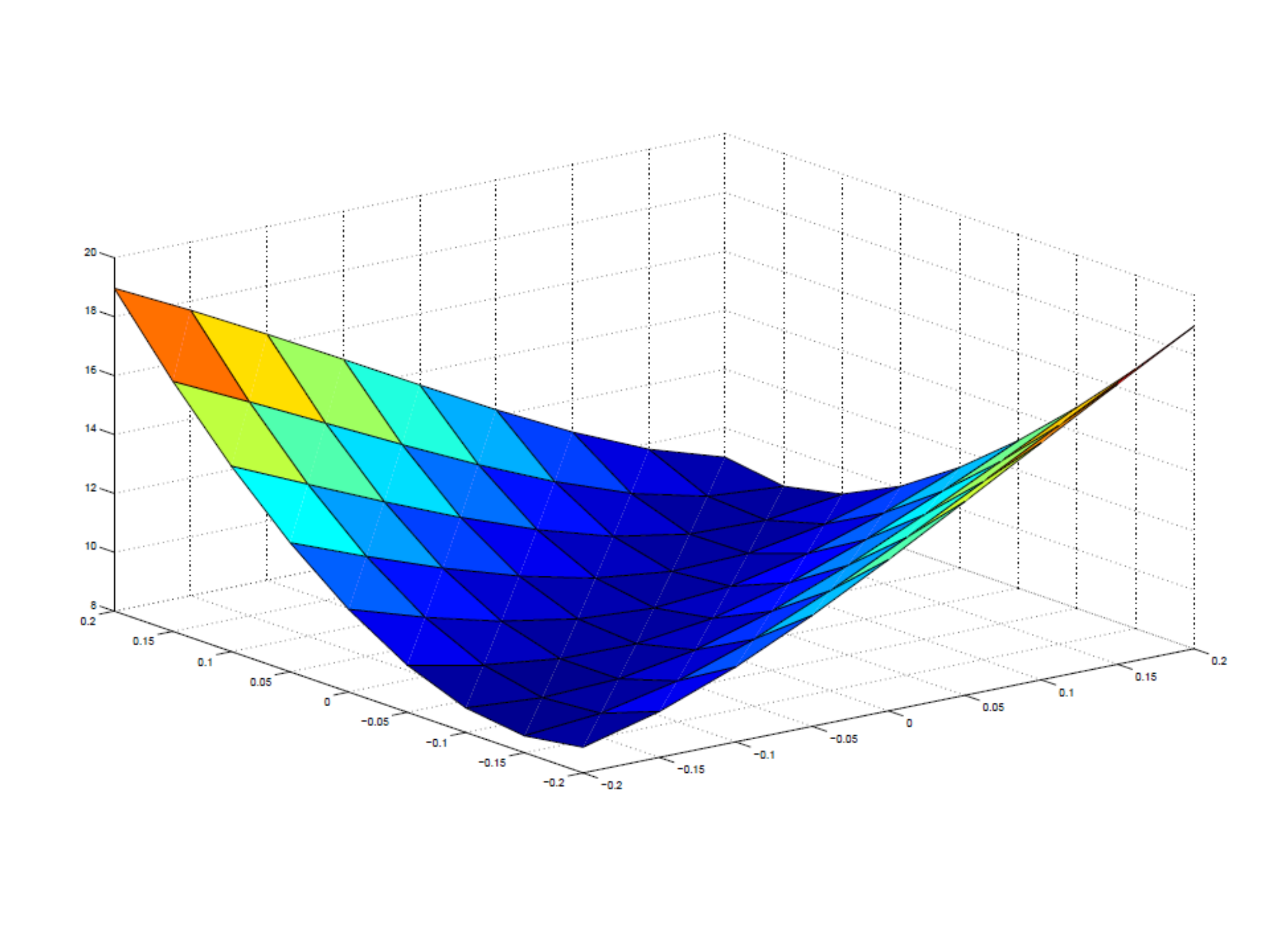}
        \caption{Price with Variation in Mean Jump Size,  the remaining parameters belong to the benchmark setting}
                \label{ch:JumpSize}
    \end{center}
\end{figure}

 We also observe that the price increases as long as the mean jump sizes are not equal, but regardless of the sign of the mean jump sizes. Table \ref{tbl:JumpSize} shows that the price is extremely sensitive to the mean jump size, and particularly to the differences in the mean jump size. This effect is magnified when we increase the parameter $\lambda$ as the number of jumps, and hence the effect of the jumps on the price increases. In these cases the slope as we leave the region where the mean jump sizes are equal tends to increase more rapidly.
\begin{table}[htb!]
\caption{Comparison of Prices for case for Variation in Mean Jump Size $\overline{k^{(m)}}$}
\label{tbl:JumpSize}
\begin{tabular}{|l|l|l|l|l|l|l|l|l|l|}
\hline
\textbf{$\overline{k^{(1)}}$/$\overline{k^{(2)}}$} & \textbf{-0.2} & \textbf{-0.15} & \textbf{-0.1} & \textbf{-0.05} & \textbf{0} & \textbf{0.05} & \textbf{0.1} & \textbf{0.15} & \textbf{0.2}  \\ \hline
\textbf{-0.2} & 8.87 & 9.56 & 10.55 & 11.77 & 13.14 & 14.58 & 16.06 & 17.52 & 18.95 \\ \hline
\textbf{-0.15} & 8.57 & 8.81 & 9.38 & 10.26 & 11.39 & 12.67 & 14.03 & 15.42 & 16.81 \\ \hline
\textbf{-0.1} & 8.83 & 8.65 & 8.77 & 9.22 & 10 & 11.03 & 12.23 & 13.51 & 14.82 \\ \hline
\textbf{-0.05} & 9.6 & 9.07 & 8.76 & 8.74 & 9.08 & 9.76 & 10.7 & 11.81 & 13.01 \\ \hline
\textbf{0} & 10.81 & 10.01 & 9.35 & 8.9 & 8.75 & 8.96 & 9.53 & 10.38 & 11.41 \\ \hline
\textbf{0.05} & 12.39 & 11.39 & 10.46 & 9.67 & 9.07 & 8.77 & 8.86 & 9.32 & 10.09 \\ \hline
\textbf{0.1} & 14.3 & 13.16 & 12.04 & 10.98 & 10.03 & 9.28 & 8.82 & 8.77 & 9.13 \\ \hline
\textbf{0.15} & 16.49 & 15.27 & 14.01 & 12.76 & 11.55 & 10.45 & 9.52 & 8.9 & 8.71 \\ \hline
\textbf{0.2} & 18.97 & 17.69 & 16.35 & 14.97 & 13.57 & 12.21 & 10.92 & 9.81 & 9.01 \\ \hline
\end{tabular}
\end{table}

\begin{figure}[htb!]
    \begin{center}
        \includegraphics[clip=false,scale=0.25]{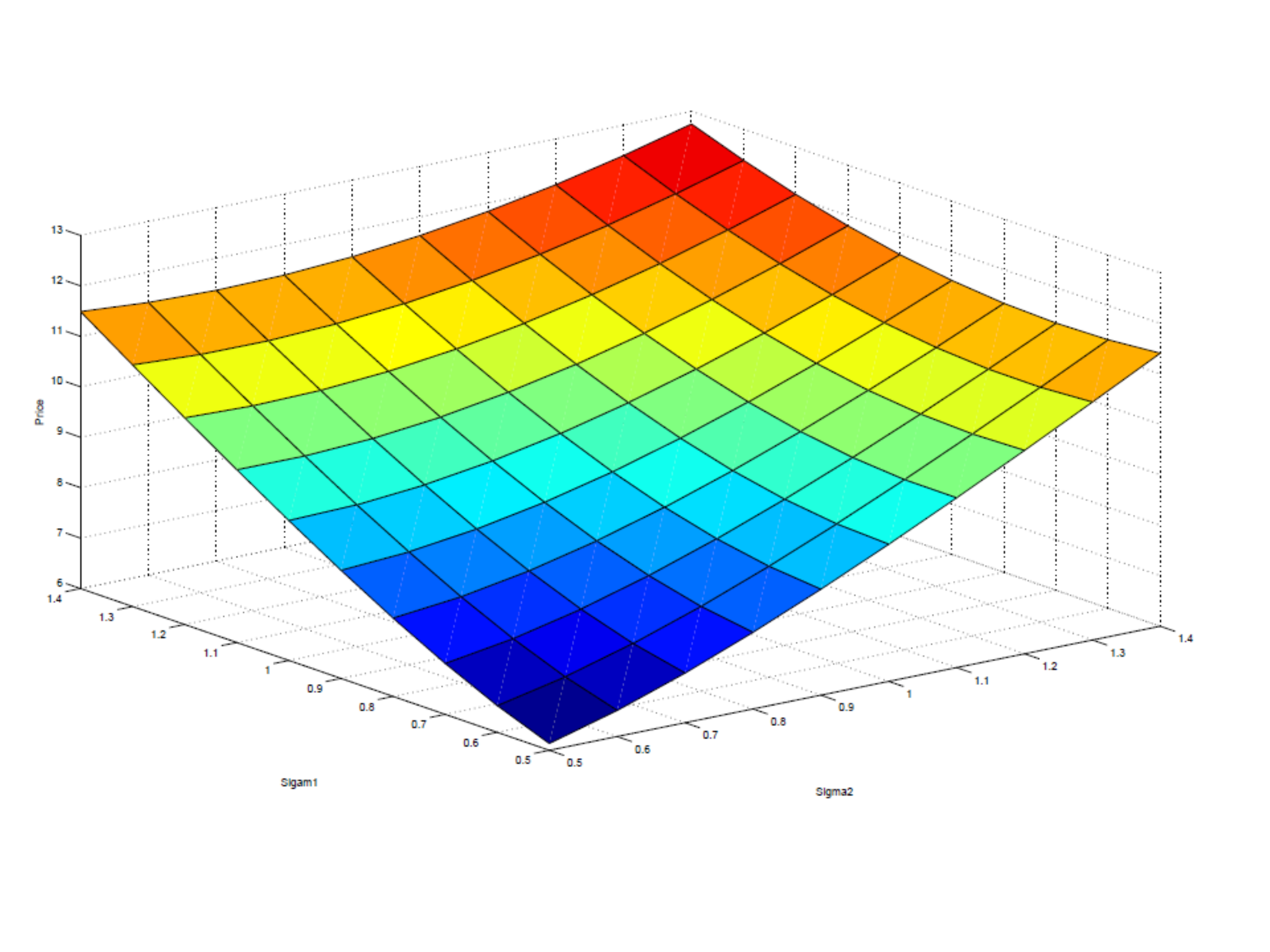}
        \caption{Price with Variation in Asset-Volatility Multiplier $\sigma^{(m)}$,  the remaining parameters belong to the benchmark setting}
        \label{ch:AssetVolMult}
    \end{center}
\end{figure}

Figure \ref{ch:AssetVolMult} shows the variation in price for different values of the asset-volatility multiplier, $\sigma^{(m)}$. For both assets we see an increase as we increase $\sigma^{(m)}$, as this tends to increase the overall volatility of each asset. The effects of $\sigma^{(m)}$ are more muted than other parameters, as shown in Table \ref{tbl:Sigma}, as they tend to linearly increase the price and the price does not seem to be more strongly related to either $\sigma^{(1)}$ or $\sigma^{(2)}$.
\begin{table}[htb!]
\caption{Comparison of Prices for case for Variation in $\sigma^{(m)}$}
\label{tbl:Sigma}
\begin{tabular}{|l|l|l|l|l|l|l|l|l|l|l|}
\hline
\textbf{$\sigma^{(1)}$/$\sigma^{(2)}$} & \textbf{0.5} & \textbf{0.6} & \textbf{0.7} & \textbf{0.8} & \textbf{0.9} & \textbf{1} & \textbf{1.1} & \textbf{1.2} & \textbf{1.3} & \textbf{1.4} \\ \hline
\textbf{0.5} & 6.13 & 6.52 & 6.98 & 7.52 & 8.1 & 8.72 & 9.36 & 10.03 & 10.72 & 11.41 \\ \hline
\textbf{0.6} & 6.53 & 6.83 & 7.21 & 7.68 & 8.2 & 8.77 & 9.37 & 10 & 10.65 & 11.32 \\ \hline
\textbf{0.7} & 7.01 & 7.23 & 7.54 & 7.93 & 8.39 & 8.9 & 9.45 & 10.04 & 10.66 & 11.3 \\ \hline
\textbf{0.8} & 7.55 & 7.7 & 7.94 & 8.26 & 8.66 & 9.11 & 9.61 & 10.15 & 10.73 & 11.33 \\ \hline
\textbf{0.9} & 8.15 & 8.23 & 8.41 & 8.67 & 9 & 9.39 & 9.84 & 10.33 & 10.86 & 11.43 \\ \hline
\textbf{1} & 8.77 & 8.81 & 8.93 & 9.13 & 9.41 & 9.74 & 10.14 & 10.58 & 11.06 & 11.59 \\ \hline
\textbf{1.1} & 9.43 & 9.42 & 9.5 & 9.65 & 9.86 & 10.15 & 10.49 & 10.88 & 11.32 & 11.8 \\ \hline
\textbf{1.2} & 10.1 & 10.06 & 10.1 & 10.2 & 10.37 & 10.6 & 10.89 & 11.24 & 11.63 & 12.07 \\ \hline
\textbf{1.3} & 10.79 & 10.72 & 10.72 & 10.78 & 10.91 & 11.1 & 11.34 & 11.64 & 11.99 & 12.39 \\ \hline
\textbf{1.4} & 11.49 & 11.4 & 11.37 & 11.39 & 11.48 & 11.63 & 11.83 & 12.09 & 12.4 & 12.75 \\ \hline
\end{tabular}
\end{table}
Finally, we examine the effect of varying the jump-size variance on the price, as shown in Figure \ref{ch:JumpVar}. As we increase the variance in the jumps for both assets we see the price increase, as we would expect, and in a non-linear fashion.
\begin{figure}[htb!]
    \begin{center}
        \includegraphics[width=\textwidth]{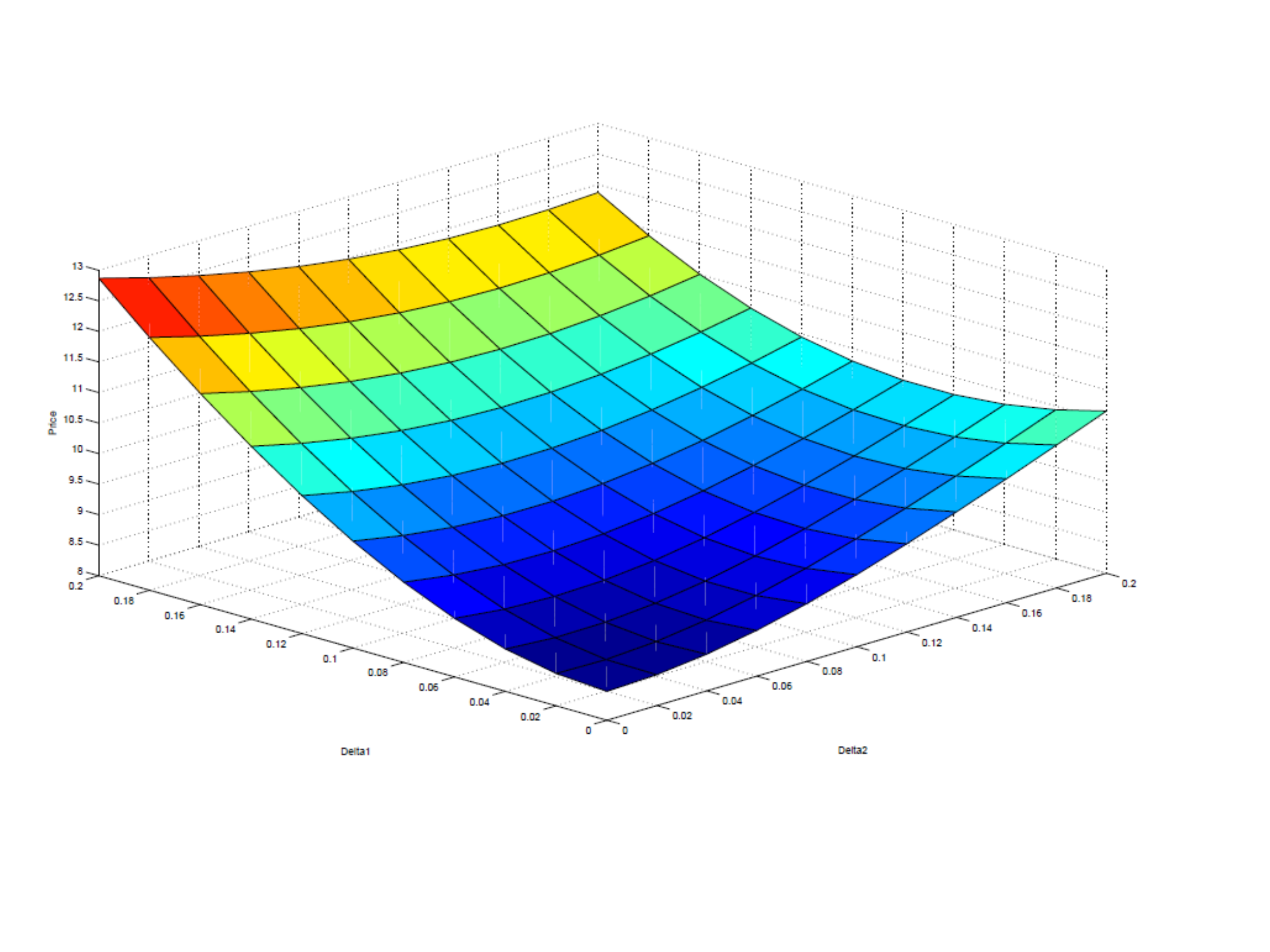}
        \caption{Price with Variation in Jump-Size Variance,  the remaining parameters belong to the benchmark setting}
               \label{ch:JumpVar}
    \end{center}
\end{figure}
Increases in the variance for the long asset tend to have a greater effect on the price, since there's a greater probability that the long asset will increase it's spread over the short asset and increasing the variance of jumps in both assets tends to compound the effects, as we would expect. Table \ref{tbl:Delta} shows that the jump-size variance parameter has a fairly substantial effect on the price, as increasing from 0.02 to 0.2 can cause an increase of 40\% in the price of the spread.
\begin{table}[htb!]
\caption{Comparison of Prices for case for Variation in Jump-Size Variance $\delta^{(m)}$}
\label{tbl:Delta}
\begin{tabular}{|l|l|l|l|l|l|l|l|l|l|l|}
\hline
\textbf{$\delta^{(1)}$/$\delta^{(2)}$} & \textbf{0.02} & \textbf{0.04} & \textbf{0.06} & \textbf{0.08} & \textbf{0.1} & \textbf{0.12} & \textbf{0.14} & \textbf{0.16} & \textbf{0.18} & \textbf{0.2} \\ \hline
\textbf{0.02} & 8.52 & 8.57 & 8.68 & 8.84 & 9.05 & 9.29 & 9.55 & 9.84 & 10.13 & 10.44 \\ \hline
\textbf{0.04} & 8.64 & 8.65 & 8.72 & 8.84 & 9.01 & 9.22 & 9.46 & 9.72 & 10 & 10.29 \\ \hline
\textbf{0.06} & 8.87 & 8.84 & 8.86 & 8.94 & 9.07 & 9.25 & 9.46 & 9.69 & 9.95 & 10.22 \\ \hline
\textbf{0.08} & 9.19 & 9.12 & 9.11 & 9.14 & 9.23 & 9.37 & 9.54 & 9.74 & 9.97 & 10.22 \\ \hline
\textbf{0.1} & 9.6 & 9.5 & 9.44 & 9.44 & 9.49 & 9.58 & 9.71 & 9.88 & 10.08 & 10.29 \\ \hline
\textbf{0.12} & 10.09 & 9.95 & 9.87 & 9.83 & 9.83 & 9.88 & 9.98 & 10.11 & 10.27 & 10.45 \\ \hline
\textbf{0.14} & 10.64 & 10.48 & 10.36 & 10.29 & 10.26 & 10.27 & 10.33 & 10.41 & 10.54 & 10.69 \\ \hline
\textbf{0.16} & 11.25 & 11.07 & 10.93 & 10.83 & 10.76 & 10.74 & 10.75 & 10.81 & 10.89 & 11 \\ \hline
\textbf{0.18} & 11.92 & 11.73 & 11.56 & 11.43 & 11.34 & 11.28 & 11.26 & 11.27 & 11.32 & 11.4 \\ \hline
\textbf{0.2} & 12.64 & 12.43 & 12.25 & 12.1 & 11.98 & 11.89 & 11.84 & 11.82 & 11.83 & 11.87 \\ \hline
\end{tabular}
\end{table}

\subsection{Effect of Discretization and Truncation}

As Hurd and Zhou  observe, the effect of varying the damping parameter $\epsilon$ on the price is relatively small. It should be noted, however that this is not universally true, and that for values of $\epsilon_2 < 0.2$, or for  $\epsilon_1 - \epsilon_2 - 1 < 0.2$ we do observe some error in the prices produced by our FFT method, as seen in Table (\ref{tbl:Eps}).
\begin{table}[htb!]
\begin{center}
\caption{Variation in Price for various choices of $\epsilon = (\epsilon_1, \epsilon_2)$}
\label{tbl:Eps}
\begin{tabular}{|l|l|l|l|l|l|l|l|l|l|}
\hline
 \textbf{$\epsilon_1 / \epsilon_2$} & \textbf{0.1} & \textbf{0.2} & \textbf{0.3} & \textbf{0.4} & \textbf{0.5} & \textbf{1.6} & \textbf{1.7} & \textbf{1.8} & \textbf{1.9} \\ \hline
\textbf{-3.0} & 10.7 & 8.81 & 8.77 & 8.77 & 8.77 & 8.77 & 8.77 & 8.78 & 9.08 \\ \hline
\textbf{-2.9} & 10.7 & 8.81 & 8.77 & 8.77 & 8.77 & 8.77 & 8.78 & 9.1 & - \\ \hline
\textbf{-2.8} & 10.7 & 8.81 & 8.77 & 8.77 & 8.77 & 8.78 & 9.11 & - & - \\ \hline
\textbf{-2.7} & 10.7 & 8.81 & 8.77 & 8.77 & 8.77 & 9.12 & 7.17E+13 & - & - \\ \hline
\textbf{-1.6} & 10.7 & 8.81 & 8.77 & 8.78 & 9.29 & - & - & - & - \\ \hline
\textbf{-1.5} & 10.7 & 8.81 & 8.78 & 9.31 & - & - & - & - & - \\ \hline
\textbf{-1.4} & 10.7 & 8.82 & 9.33 & - & - & - & - & - & - \\ \hline
\textbf{-1.3} & 10.71 & 9.39 & - & - & - & - & - & - & - \\ \hline
\textbf{-1.2} & 11.34 & 1.20E+15 & - & - & - & - & - & - & - \\ \hline
\end{tabular}
\end{center}
\end{table}

We can also test the sensitivity of our model to variations in the number of steps and the step size. Note that since $\bar{u} = \frac{N \eta}{2}$ for a fixed N if we attempt to decrease the truncation error by increasing $\bar{u}$ (or, in the case of our optimal step size algorithm, $\bar{u}_{min}$), our discretization error will increase, as the step size integration interval is proportional to the step-size. Fortunately, however, our method shows very little sensitivity to either the step-size or the truncation interval, as displayed in Table \ref{tbl:NStep}.
\begin{table}[htb!]
\begin{center}
\caption{Variation in Prices with change in $N$ and $\bar{u}_{min}$}
\label{tbl:NStep}
\begin{tabular}{|l|l|l|l|l|}
 \hline
\textbf{$\bar{u}_{min}$ / $N$} & \textbf{128} & \textbf{256} & \textbf{512} & \textbf{1024} \\ \hline
\textbf{40} & 8.777635 & 8.772048 & 8.772048 & 8.772048 \\ \hline
\textbf{60} & - & 8.772241 & 8.772048 & 8.772048 \\ \hline
\textbf{80} & - & 8.777052 & 8.772048 & 8.772048 \\ \hline
\textbf{120} & - & - & 8.772060 & 8.772048 \\ \hline
\textbf{140} & - & - & 8.772223 & 8.772048 \\ \hline
\textbf{160} & - & - & 8.773194 & 8.772048 \\
  \hline
\end{tabular}
\end{center}
\end{table}

\section{Conclusion}
This paper has extended the model of Bates(1996) to investigate two market multivariate market models with both jumps and stochastic volatility, and derived the characteristic function under each model. Using Fourier transform techniques and the implementation of Hurd and Zhou(2009)  we have been able to produce results which very closely matched those produced by Monte Carlo methods in a fraction of the computation time.

As expected, we saw that the prices produced under our models were very sensitive to both the jump and correlation parameters. Having these additional components in our models gives us powerful tools to mold our model through parameter variation to better reflect certain characteristics observed in financial markets such as volatility and correlation smiles and smirks. On the other hand, this underscores the need for development of model calibration tools, which is not a straightforward task in the presence of jumps.

\end{document}